\documentclass[letterpaper,american,USenglish]{lipics-v2018}
	\def\version{arxiv}
	\usepackage{hyperref}

\newcommand\droppedForArXiv[1]{}

\def\draftmode{false}

\usepackage[utf8]{inputenc}
\input{preamble}

\ifdraft{}{
	\usepackage{microtype}
}

\ifarxiv{
	\hideLIPIcs
	\nolinenumbers
}{}

\usepackage{hyperref}

\hypersetup{hidelinks}

\title{Entropy Trees and Range-Minimum Queries In Optimal Average-Case Space}

\titlerunning{Entropy Trees and RMQ In Optimal Average-Case Space}

\author{J.\ Ian Munro}{University of Waterloo, Canada}{imunro@uwaterloo.ca}{https://orcid.org/0000-0002-7165-7988}{}
\author{Sebastian Wild}{University of Waterloo, Canada}{wild@uwaterloo.ca}{https://orcid.org/0000-0002-6061-9177}{}

\authorrunning{J.\,I. Munro, and S. Wild}
\Copyright{J.\ Ian Munro, and Sebastian Wild}

\subjclass{%
}


\begin{document}

\maketitle

\paragraph{Abstract}

The range-minimum query (RMQ) problem is a fundamental data structuring task
with numerous applications.
Despite the fact that succinct solutions with worst-case optimal $2n+o(n)$
bits of space and constant query time are known,
it has been unknown whether such a data structure can be made adaptive
to the reduced entropy of \emph{random} inputs (Davoodi et al.\ 2014).
We construct a succinct data structure with the optimal $1.736n+o(n)$ bits of space
on average for random RMQ instances, settling this open problem.

Our solution relies on a compressed data structure for binary trees that is 
of independent interest. 
It can store a (static) binary search tree generated by random insertions
in asymptotically optimal expected space and supports many queries in constant time.
Using an instance-optimal encoding of subtrees, we furthermore obtain a ``hyper-succinct''
data structure for binary trees that improves upon the ultra-succinct representation
of Jansson, Sadakane and Sung (2012).

\section{Introduction}

The range-minimum query (RMQ) problem is the following data structuring task:
Given an array $A[1..n]$ of comparable items, construct a data structure 
at preprocessing time that can answer subsequent queries without inspecting $A$ again.
The answer to the query $\mathit{rmq}(i,j)$, for $1\le i\le j\le n$, is the index (in $A$) of the%
\footnote{%
	To simplify the presentation, we assume the elements in $A$ are unique.
	In the general case, we fix a tie-breaking rule, usually to return the leftmost minimum.
	Our data structures extend to any such convention.%
}
minimum in $A[i..j]$, \ie,
\[
		\mathit{rmq}(i,j) 
	\wwrel=
		\mathop{\arg\min}\limits_{i\le k\le j} A[k].
\]
RMQ data structures are fundamental building blocks
to find lowest common ancestors in trees,
to solve the longest common extension problem on strings,
to compute suffix links in suffix trees,
they are used as part of compact suffix trees, 
for (3-sided) orthogonal range searching,
for speeding up document retrieval queries,
finding maximal-scoring subsequences,
and they can be used to compute Lempel-Ziv-77 factorizations given only the suffix array
(see \wref{sec:applications} and \cite[\S3.3]{Fischer2007} for details).
Because of the applications and its fundamental nature, 
the RMQ problem has attracted significant attention in the data structures community, 
in particular the question about solutions with smallest 
possible space usage (succinct data structures).

The hallmark of this line of research is the work of Fischer and Heun~\cite{FischerHeun2011}
who describe a data structure that uses $2n+o(n)$ bits of space and answers queries in constant
time.%
\footnote{%
	We assume here, and in the rest of this article, that we are working on a 
	word-RAM with word size $w = \Theta(\log n)$.%
}
We discuss further related work in \wref{sec:related-work}.

The space usage of Fischer and Heun's data structure is asymptotically optimal \emph{in the worst case} 
in the encoding model (\ie, when $A$ is not available at query time):
the sets of answers to range-minimum queries over arrays of length $n$ is in 
bijection with binary trees with $n$ nodes~\cite{GabowBentleyTarjan1984,Vuillemin1980},
and there are $2^{2n-\Theta(\log n)}$ such binary trees.
We discuss this bijection and its implications for the RMQ problem in detail below.

In many applications, not all possible sets of RMQ answers (resp.\ tree shapes) are possible 
or equally likely to occur. Then, more space-efficient RMQ solutions are possible.
A natural model is to consider arrays containing a random permutation;
then the effective entropy for encoding RMQ answers is asymptotically 
$1.736n$~\cite{GolinIaconoKrizancRamanRaoShende2016} instead of the $2n$ bits.
It is then natural to ask for a range-minimum data structure that uses $1.736n$ bits 
\emph{on average} in this case,
and indeed, this has explicitly been posed as an 
open problem by Davoodi, Navarro, Raman and Rao~\cite{DavoodiNavarroRamanRao2014}.
They used the ``ultra-succinct trees'' of 
Jansson, Sadakane and Sung \cite{JanssonSadakaneSung2012},
to obtain an RMQ data structure with $1.919n + o(n)$ bits on average~\cite{DavoodiNavarroRamanRao2014}.

In this note, we close the space gap and present a data structure that uses the optimal
$1.736n + o(n)$ bits on average to answer range-minimum queries for arrays whose
elements are randomly ordered.
The main insight is to base the encoding on ``subtree-size'' distributions
instead of the node-degree distribution used in the ultra-succinct trees.
To our knowledge, this is the first data structure that exploits non-uniformity of 
a \emph{global} property of trees (sizes of subtrees)~-- as opposed to 
local parameters (such as node degrees)~--
and this extended scope is necessary for optimal compression.

To obtain a data structure with constant-time queries, we modify the tree-covering
technique of Farzan and Munro~\cite{FarzanMunro2012} 
(see \wref{sec:tree-covering})
to use a more efficient encoding for micro trees.
Finally, we propose a ``hyper-succinct'' data structure for trees 
that uses an instance-optimal encoding for the micro trees.
By asymptotic average space, this data structure attains the limit of compressibility
achievable with tree-covering.

The rest of this paper is organized as follows.
We summarize applications of RMQ, its relation to lowest common ancestors
and previous work in the remainder of this first section.
In \wref{sec:preliminaries}, we introduce notation and preliminaries.
In \wref{sec:random-rmq-random-bst}, we define our model of inputs, for which 
\wref{sec:lower-bound} reports a space lower bound.
In \wref{sec:encoding}, we present an optimal encoding of binary trees
\wrt to this lower bound.
We review tree covering in \wref{sec:tree-covering} and modify it in \wref{sec:entropy-trees}
to use our new encoding.
Hyper-succinct trees are described in \wref{sec:hyper-succinct-trees}.
\droppedForArXiv{
We conclude with a discussion of our results in \wref{sec:conclusion}.
}

\subsection{Applications}
\label{sec:applications}

The RMQ problem is an elementary building block in many data structures.
We discuss two exemplary applications here, in which a non-uniform distribution
over the set of RMQ answers is to be expected.

\paragraph{Range searching}
A direct application of RMQ data structures lies in 3-sided orthogonal 2D range searching.
Given a set of points in the plane, we maintain an array of the points sorted by $x$-coordinates
and build a range-minimum data structure for the array of $y$-coordinates
and a predecessor data structure for the set of $x$-coordinates.
To report all points in $x$-range $[x_1,x_2]$ and $y$-range $(-\infty, y_1]$,
we find the indices $i$ and $j$ of the outermost points enclosed in $x$-range, \ie, 
the ranks of the successor of $x_1$ resp.\ the predecessor of $x_2$,
Then, the range-minimum in $[i,j]$ is the first candidate, 
and we compare its $y$-coordinate to $y_1$.
If it is smaller than $y_1$, we report the point and recurse in both subranges;
otherwise, we stop.

A natural testbed is to consider random point sets.
When $x$- and $y$-coordinates are independent of each other,
the ranking of the $y$-coordinates of points sorted by $x$ form a random permutation, 
and we obtain the exact setting studied in this paper.

\paragraph{Longest-common extensions}
A second application of RMQ data structures is the longest-common extension (LCE) problem on strings:
Given a string/text $T$, the goal is to create a data structure that allows to answer LCE queries,
\ie, given indices $i$ and $j$, what is the largest length $\ell$, so that 
$T_{i,i+\ell-1} = T_{j,j+\ell-1}$.
LCE data structures are a building block, \eg, for finding tandem repeats in genomes;
(see Gusfield's book~\cite{Gusfield1997} for many more applications).

A possible solution is to compute
the suffix array $\mathit{SA}[1..n]$, its inverse $\mathit{SA}^{-1}$, and 
the longest common prefix array $\mathit{LCP}[1..n]$ for the text $T$, where 
$\mathit{LCP}[i]$ stores the length of the longest common prefix of the $i$th and $(i-1)$st 
suffixes of $T$ in lexicographic order.
Using an RMQ data structure on $\mathit{LCP}$, $\mathit{lce}(i,j)$ is found as 
$\mathit{LCP}\bigl[\mathit{rmq}_\mathit{LCP}\bigl(\mathit{SA}^{-1}(i),\mathit{SA}^{-1}(j)\bigr) \bigr]$.
%

Since LCE effectively asks for lowest common ancestors of leaves in suffix trees, 
the tree shapes arising from this application are related to the shape of the suffix tree of $T$.
This shape heavily depends on the considered input strings, but
for strings generated by a Markov source, it is known that 
random suffix trees behave asymptotically similar to random tries constructed
from independent strings of the same source~\cite[Cha.\,8]{JacquetSzpankowski2015}.
Those in turn have logarithmic height~-- as do random BSTs.
This gives some hope that the RMQ instances arising in the LCE
problem are compressible by similar means.
Indeed, we could confirm the effectiveness of our compression methods
on exemplary text inputs\droppedForArXiv{, see \wref{sec:compressibility-examples}}.

\subsection{Range-minimum queries, Cartesian trees and lowest common ancestors}
\label{sec:rmq-cartesian-tree-lca}

Let $A[1..n]$ store the numbers $x_1,\ldots,x_n$,
\ie, $x_j$ is stored at index $j$ for $1\le j\le n$.
%
The Cartesian tree $T$ for $x_1,\ldots,x_n$ (resp.\ for $A[1..n]$) 
is a binary tree defined recursively as follows:
If $n=0$, it is the empty tree (``null''). 
Otherwise it is a root whose left child is the Cartesian tree for $x_1,\ldots,x_{j-1}$
and its right child is the Cartesian tree for $x_{j+1},\ldots,x_n$ where
$j$ is the position of the minimum, $j = \mathop{\arg\min}_k A[k]$;
see \wpref{fig:example-bst} for an example.
A classic observation of Gabow et al.~\cite{GabowBentleyTarjan1984} is that 
range-minimum queries on $A$ are isomorphic to lowest-common-ancestor (LCA) queries on $T$
when identifying nodes with their inorder index:
\[
		\mathit{rmq}_A(i,j) 
	\wwrel= 
		\mathit{noderank}_{\mathit{inorder}}\Bigl( 
			\mathit{lca}_T\bigl(
				\mathit{nodeselect}_{\mathit{inorder}}(i),
				\mathit{nodeselect}_{\mathit{inorder}}(j)
			\bigr)
		\Bigr).
\]
We can thus reduce an RMQ instance (with arbitrary input) to an LCA instance
of the same size (the number of nodes in $T$ equals the length of the array).
A widely-used reduction from LCA on arbitrary (ordinal) trees 
writes down the depths of nodes in an Euler tour of the tree.
This produces an RMQ instance of length $2n$ for a tree with $n$ nodes.
However, when $T$ is a \emph{binary} tree, we can replace the Euler tour by a simple 
inorder traversal and thus obtain an RMQ instance of the \emph{same} size.

\subsection{Related Work}
\label{sec:related-work}

Worst-case optimal succinct data structures for the RMQ problem have been presented by 
Fischer and Heun~\cite{FischerHeun2011},
with subsequent simplifications by Ferrada and Navarro~\cite{FerradaNavarro2017}
and Baumstark et al.~\cite{BaumstarkGogHeuerLabeit2017}.
Implementations of (slight variants) of these solutions are part of
widely-used programming libraries for succinct data structures, 
such as Succinct~\cite{succinct} and SDSL~\cite{GogBellerMoffatPetri2014}.

The above RMQ data structures make implicit use of the connection to LCA queries in trees,
but we can more generally formulate that task as a problem on trees:
Any (succinct) data structure for binary trees that supports 
finding nodes by inorder index ($\mathit{nodeselect}_{\mathit{inorder}}$),
computing LCAs, and
finding the inorder index of a node ($\mathit{noderank}_{\mathit{inorder}}$)
immediately implies a (succinct) solution for RMQ.
Most literature on succinct data structures for trees has focused on \emph{ordinal} trees,
\ie, trees with unbounded degree where only the order of children matters, 
but no distinction is made, \eg, between a left and right single child.
Some ideas can be translated to \emph{cardinal} trees 
(and binary trees as a special case thereof)~\cite{FarzanMunro2012,DavoodiRamanRao2017}.
For an overview of ordinal-tree data structures
see, \eg, the survey of Raman and Rao~\cite{RamanRao2013} or Navarro's book~\cite{Navarro2016}.
From a theoretical perspective, the tree-covering technique~-- initially suggested
by Geary, Raman and Raman~\cite{GearyRamanRaman2006}; extended and 
simplified in~\cite{HeMunroRao2012,FarzanMunro2012,DavoodiNavarroRamanRao2014}~--
is the most versatile and expressive representation.
We present the main results and techniques, 
restricted to the subset of operations and adapted to binary trees,
in \wref{sec:tree-covering}.

A typical property of succinct data structures is that their space usage is determined only by the
\emph{size} of the input. 
For example, all of the standard tree representations use $2n+o(n)$ bits of space
for \emph{any} tree with $n$ nodes, and 
the same is true for the RMQ data structures cited above.
This is in spite of the fact that for certain distributions over possible tree shapes,
the entropy can be much lower than $2n$ bits~\cite{KiefferYangSzpankowski2009}.

There are a few exceptions.
For RMQ, Fischer and Heun~\cite{FischerHeun2011} show that range-minimum queries can still be
answered efficiently when the array is compressed to $k$th order empirical entropy.
For random permutations, the model studied in this article, 
this does not result in significant savings, though.
Barbay, Fischer and Navarro~\cite{BarbayFischerNavarro2012} used LRM-trees to
obtain an RMQ data structure that adapts to presortedness in $A$, \eg, the number of (strict) runs.
Again, for the random permutations considered here, this would not result in space reductions.
Recently, Jo, Mozes and Weimann~\cite{JoMozesWeimann2018} designed RMQ solutions 
for grammar-compressed input arrays resp.\ DAG-compressed Cartesian trees.
The amount of compression for random permutation is negligible for the former; 
for the latter it is less clear, but in both cases, they have to give up constant-time queries.

In the realm of ordinal trees, Jansson, Sadakane and Sung~\cite{JanssonSadakaneSung2012}
designed an ``ultra-succinct'' data structure by
replacing the unary code for node degrees in the DFUDS representation by an encoding 
that adapts to the \emph{distribution} of node degrees.
Davoodi et al.~\cite{DavoodiNavarroRamanRao2014} used the same encoding in tree covering
to obtain the first ultra-succinct encoding for binary trees with inorder support.
They show that for random RMQ instances, a node in the Cartesian tree has probability $\frac13$
to be binary resp.\ a leaf, and probability $\frac16$ to have a single left resp.\ right child.
The resulting entropy is $\mathcal H(\frac13,\frac13,\frac16,\frac16) \approx 1.91$ bit per node
instead of the $2$ bit for a trivial encoding.

Golin et al.~\cite{GolinIaconoKrizancRamanRaoShende2016} showed that
$1.736n$ bits are (asymptotically) necessary and sufficient to encode a random RMQ instance,
but they do not present a data structure that is able to make use of the encoding.
The constant in the lower bound also appears in the entropy of BSTs build from random 
insertions~\cite{KiefferYangSzpankowski2009}, for reasons that will become obvious in \wref{sec:random-rmq-random-bst}.
Similarly, the encoding of Golin et al.\ has independently been described by 
Magner, Turowski and Szpankowski~\cite{MagnerTurowskiSzpankowski2018}
to compress trees (without attempts to combine it with efficient access to the stored object).
There is thus a gap left between the lower bound and 
the best data structure with efficient queries, both for RMQ and for representing binary trees.

\section{Notation and Preliminaries}
\label{sec:preliminaries}

We write $[n..m] = \{n,\ldots,m\}$ and $[n] = [1..n]$ for integers $n$, $m$.
We use $\lg$ for $\log_2$ and leave the basis of $\log$ undefined (but constant);
(any occurrence of $\log$ outside an Landau-term should thus be considered a mistake).
$\mathcal T_n$ denotes the set of binary trees on $n$ nodes, \ie, every node has 
a left and a right child (both potentially empty / null).
For a tree $t\in\mathcal T_n$ and one of its nodes $v\in t$, we write
$\mathit{st}_t(v)$ for the subtree size of $v$ in $t$, \ie, the number of nodes $w$
(including $w=v$)
for which $v$ lies on the path from the root to $w$.
When the tree is clear form the context, we shortly write $\mathit{st}(v)$.

\subsection{Bit vectors}

We use the data structure of 
Raman, Raman, and Rao~\cite{RamanRamanRao2007}
for compressed bitvectors. 
They show the following result;
we use it for two more specialized data structures below.

\begin{lemma}[Compressed bit vector]
\label{lem:compressed-bit-vectors}
	Let $\mathcal{B}$ be a
	bit vector of length $n$, containing $m$ $1$-bits.  In the
	word-RAM model with word size $\Theta(\lg n)$ bits, there is a data
	structure of size 
	\begin{align*}
			\lg \binom{n}{m} \wbin+ O\biggl(\frac{n\log \log n}{\log n}\biggr)
		&\wwrel\le 
			n H\Bigl(\frac{m}{n}\Bigr) \wbin+ O\biggl(\frac{n\log \log n}{\log n}\biggr) 
	\\	&\wwrel= 
			m \lg \Bigl(\frac nm\Bigr) \wbin+ O\biggl(\frac{n \log \log n}{\log n}+m\biggr)
	\end{align*}
	bits that
	supports the following operations in $O(1)$ time, 
	for any $i \in [1,n]$:
	\begin{itemize}
		\item $\mathit{access}(\mathcal{V}, i)$: return the bit at index $i$ in $\mathcal{V}$.
		\item $\mathit{rank}_\alpha(\mathcal{V}, i)$: return the number of bits with
		value $\alpha \in \{0,1\}$ in $\mathcal{V}[1..i]$.
		\item $\mathit{select}_\alpha(\mathcal{V}, i)$: return the index of the $i$-th
		bit with value $\alpha \in \{0,1\}$.
	\end{itemize}
\end{lemma}

\subsection{Variable-cell arrays}

Let $o_1,\ldots, o_m$ be $m$ objects
where $o_i$ needs $s_i$ bits of space.
The goal is to store
an ``array'' $O$ of the objects contiguously in memory, so that
we can access the $i$th element in constant time as $O[i]$;
in case $s_i > w$, we mean by ``access'' to find its starting position.
We call such a data structure a variable-cell array.

\begin{lemma}[Variable-cell arrays]
\label{lem:variable-cell-arrays}
	There is a variable-cell array data structure for 
	objects $o_1,\ldots, o_m$ of sizes $s_1,\ldots,s_m$
	that occupies
	\[\sum s_i \bin+ m \lg(\max s_i) \bin+ 2 m \lg \lg n \bin+ O(m)\]
	bits of space.
\end{lemma}

\begin{proof}
Let $n = \sum s_i$ be the total size of all objects and
denote by $s = \min s_i$, $S = \max s_i$ and $\bar s = n/m$
the minimal, maximal and average size of the objects, respectively.
We store the concatenated bit representation in a bitvector $B[1..n]$
and use a two-level index to find where the $i$th object begins.

Details:
Store the starting index of every $b$th object
in an array $\mathit{blockStart}[1..\lceil m/b\rceil]$.
The space usage is $\frac mb \lg n$.
In a second array $\mathit{blockLocalStart}[1..m]$,
we store for every object its starting index within its block.
The space for this is $m \lg(b S)$:
we have to prepare for the worst case of a block full of maximal objects.

It remains to choose the block size;
$b = \lg^2 n$ yields the claimed bounds.
Note that $\mathit{blockStart}$ is $o(n)$ 
(for $b = \omega(\lg n / \bar s)$), but
$\mathit{blockLocalStart}$ has, in general, non-negligible space overhead.
The error term only comes from ignoring ceilings around the logarithms;
its constant can be bounded explicitly.
\end{proof}

\subsection{Compressed piecewise-constant arrays}


Let $A[1..n]$ be a static array of objects of size $w$ bits each.
The only operation is the standard read-access $A[i]$
where $i\in[n]$.
Let $m$ be the number of indices $i$
with $A[i] \ne A[i-1]$, \ie, the number of times we see the value in $A$ \emph{change}
in a sequential scan.
We always count $i=1$ as a change, so $m\ge 1$.
Between two such change indices, the value in $A$ is constant.
We hence call such arrays piecewise-constant arrays.
We can store such arrays in compressed form.

\begin{lemma}[Compressed piecewise-constant array]
\label{lem:piecewise-constant-array}
	Let $A[1..n]$ be a static array with $m$ value changes.
	There is a data structure for storing $A$ that allows 
	constant-time read-access to $A[i]$ using
	$m w + m \lg \frac nm + O \bigl(\frac{n \lg\lg n}{\lg n} + m\bigr)$ bits of space.
\end{lemma}

\begin{proof}
	We store an array $V[1..m]$ of the distinct values in the order they appear in $A$,
	and a bitvector $C[1..n]$ where $C[i] = 1$ iff $A[i] \ne A[i-1]$.
	We always set $C[1] = 1$.
	Since $C$ has $m$ ones, we can store it in compressed form using
	\wref{lem:compressed-bit-vectors}.
	Then, $A[j]$ is given by $V[\mathit{rank_1}(C,j)]$, which can be found in $O(1)$.
	The space is $mw$ for $V$ and $m \lg \frac nm + O(m + n \lg \lg n / \lg n)$ for~$C$.
\end{proof}

This combination of a sparse (compressed) bitvector for changes
and an ordinary (dense) vector for values was used in previous work, \eg,
\cite{GearyRamanRaman2006,DavoodiNavarroRamanRao2014}, without giving it a name. 
We feel that a name helps reduce the conceptual and notational complexity of later constructions.

\begin{remark}[Further operations]
\label{rem:piecewise-constant-further-ops}
	Without additional space, the compressed piecewise-constant array data structure can 
	also answer the following query in constant time using rank and select on the bitvector:
	$\mathit{runlen}(i)$, the number of entries to the left of index $i$ that contain the same value;
		$\mathit{runlen}(i) = \max\{\ell \ge 1 : \forall i-\ell < j < i : A[j] = A[i] \}$.
\end{remark}

\section{Random RMQ and random BSTs}
\label{sec:random-rmq-random-bst}

We consider the random permutation model for RMQ: Every (relative) ordering of
the elements in $A[1..n]$ is considered to occur with the same probability.
Without loss of generality, we identify these $n$ elements with their rank,
\ie, $A[1..n]$ contains a random permutation of $[1..n]$.
We refer to this as a random RMQ instance.

Let $T_n\in\mathcal T_n$ be (the random shape of) the Cartesian tree 
associated with a random RMQ instance (recall \wref{sec:rmq-cartesian-tree-lca}).
We will drop the subscript $n$ when it is clear form the context.
We can precisely characterize the distribution of $T_n$:
Let $t\in\mathcal T_n$ be a given (shape of a) binary tree,
and let $i$ be the inorder index of the root of $t$.
The minimum in a random permutation is located at every position $i\in[n]$ with 
probability $\frac 1n$.
Apart from renaming, the subarrays $A[1..i-1]$ and $A[i+1..n]$ contain a random permutation
of $i-1$ resp. $n-i$ elements, and these two permutations are independent of each other
conditional on their sizes.
We thus have
\begin{align}
\label{eq:prob-Tn=t-rec}
		\Prob{T_n = t}
	&\wwrel=
		\begin{dcases}
			1, & n \le 1; \\
			\tfrac1n \cdot \Prob{T_{i-1} = t_\ell} \cdot \Prob{T_{n-i} = t_r},
				& n \ge 2,
		\end{dcases}
\end{align}
where $t_\ell$ and $t_r$ are the left resp.\ right subtrees of the root of $t$.
Unfolding inductively yields
\begin{align}
\label{eq:prob-Tn=t}
		\Prob{T_n = t}
	&\wwrel=
		\prod_{v\in t} \frac1{\mathit{st}(v)},
\end{align}
%
where the product is understood to range over all nodes $v$ in $t$.
Recall that $\mathit{st}(v)$ denotes the subtree size of $v$.
The very same distribution over binary trees also arises for (unbalanced) 
binary search trees (BSTs) when they are build by successive insertions 
from a random permutation (``random BSTs'').
Here, the root's inorder rank is not given by the position of the minimum in the input,
but by the rank of the first element.
That means, the Cartesian tree for a permutation is precisely 
the BST generated by successively inserting the \emph{inverse permutation.}
Since the inverse of a random permutation is itself
distributed uniformly at random over all permutations, 
the resulting tree-shape distributions are the same.

\section{Lower bound}
\label{sec:lower-bound}

Since the sets of answers to range-minimum queries is in bijection with Cartesian trees,
the entropy $H_n$ of the distribution of the shape of the Cartesian tree 
gives an information-theoretic lower bound 
for the space required by any RMQ data structure in the encoding model.
For random RMQ instances, we find from \weqref{eq:prob-Tn=t-rec} and the 
decomposition rule of the entropy that
$H_n$ fulfills the recurrence
\begin{align}
		H_0 \wrel= H_1
	&\wwrel=
		0
\\
		H_n
	&\wwrel=
		\lg n + \frac 1n \sum_{i=1}^n (H_{i-1}+H_{n-i})
		,\qquad(n\ge 2).
\label{eq:entropy-rec}
\end{align}
We decompose the entropy of the choice of the entire tree shape into 
first choosing the root's rank (entropy $\lg n$ since we uniformly choose between $n$ outcomes)
and adding the entropy for choosing the subtrees conditional on the given subtree sizes.
The above recurrence is very related with recurrences for random binary search trees or 
the average cost in quicksort, only with a different ``toll function''.
Kieffer, Yan and Szpankowski~\cite{KiefferYangSzpankowski2009} show%
\footnote{%
	Hwang and Neininger~\cite{HwangNeininger2002} showed earlier that \weqref{eq:entropy-rec} can be solved 
	exactly for arbitrary toll function, and $\lg n$ is one such.%
}
that it
solves to 
\begin{align*}
		H_n 
	&\wwrel= 
		\lg(n) + 2(n+1) \sum_{i=2}^{n-1} \frac{\lg i}{(i+2)(i+1)}
\\	&\wwrel\sim
		2n\sum_{i=2}^{\infty} \frac{\lg i}{(i+2)(i+1)}
\\	&\wwrel\approx 1.7363771 n 
\end{align*}
(Note that Kieffer et al.~use $n$ for the number of external leaves, where we count (internal) nodes, 
so there is an off-by-one in the meaning of $n$.)
The asymptotic approximation is actually a fairly conservative upper bound for small $n$, see
\wref{fig:entropy}.

\begin{figure}
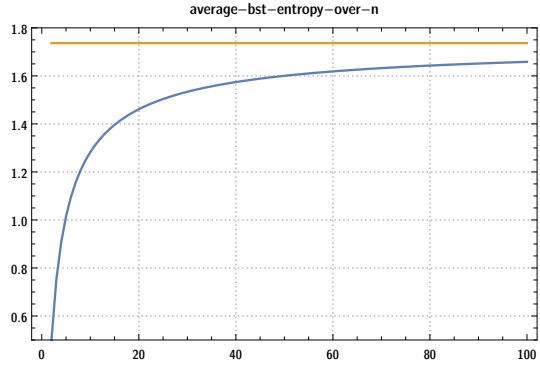

	\mathematicaPlotPdf{average-bst-entropy-over-n}
	\caption{%
		$H_n/n$ (blue) and its limit $\lim_{n\to\infty} H_n/n \approx 1.736$ (yellow) for $n\le 100$.
	}
\label{fig:entropy}
\end{figure}

\section{An optimal encoding: subtree-size code}
\label{sec:encoding}

The formula for $\Prob{T_n=t}$, \weqref{eq:prob-Tn=t},
immediately suggests a route for an optimal encoding:
We can rewrite the lower bound, $H_n$, as
\begin{align*}
	H_n
	&\wwrel=
		- \sum_{T\in\mathcal T_n} \Prob T \cdot \lg(\Prob T)
\\	&\wwrel=
		\sum_{T\in\mathcal T_n} \Prob T \cdot 
			\underbrace{ \sum_{v\in t}\lg \bigl( \mathit{st}(v) \bigr) }
				_ {\textstyle \mathcal H_{\mathit{st}(t)}}
\\[-3ex]	&\wwrel=
		\E{\mathcal H_{\mathit{st}}(T)}.
\numberthis\label{eq:Hn=E-HstT}
\end{align*}
That means, an encoding that spends $\mathcal H_{\mathit{st}}(t)$ bits to 
encode tree $t\in\mathcal T_n$
has optimal expected code length!
In a slight abuse of the term, 
we will call the quantity $\mathcal H_{\mathit{st}} (t) = \sum_{v\in t} \lg\bigl(\mathit{st}(v)\bigr)$ the 
\emph{subtree-size entropy} of a binary tree $t$.

Expressed for individual nodes, \weqref{eq:Hn=E-HstT} says
\textsl{we may spend $\lg s$ bits for a node whose subtree size is $s$.}
Now, what can we store in $\lg s$ bits for each node that uniquely describes the tree?
One option is the size of each node's \emph{left subtree}.
If a node $v$ has subtree size $\mathit{st}(v) = s$, then its
left subtree has size $\mathit{ls}(v) = \mathit{st}(\mathit{left}(v)) \in [0..s-1]$ (``left size''),
a quantity with $s$ different values.
Moreover, for random BSTs, each if these $s$ values is equally likely.

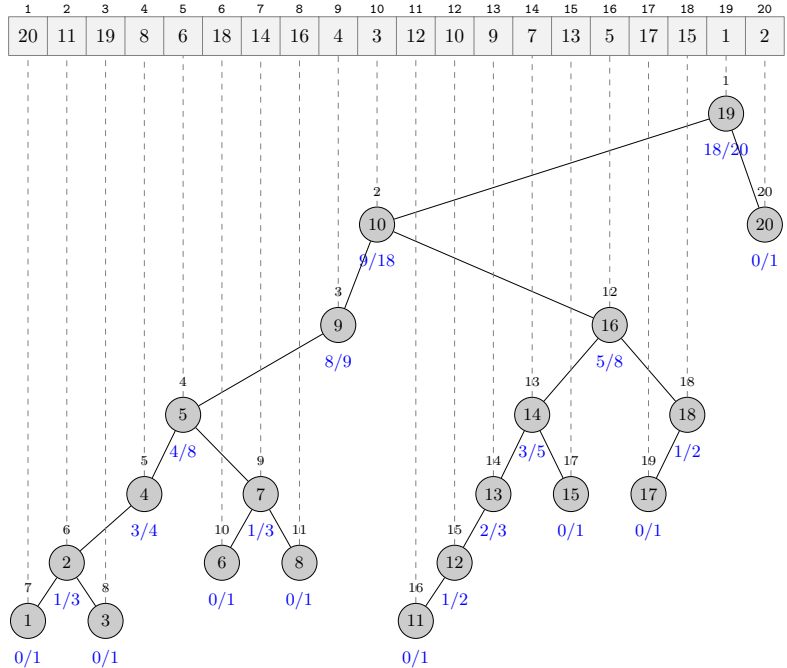
\begin{figure}[tbh]

	\plaincenter{\resizebox{.75\textwidth}!{
	\begin{tikzpicture}[
			scale=.7,
			tree node/.style         = {circle, draw, fill=black!20, font=\small, minimum size=18pt, inner sep=0pt},
			preorder label/.style    = {font=\scriptsize},
			subtreesize label/.style = {font=\small, text=blue},
			tree edge/.style         = {thin},
			dfs block/.style         = {draw=blue!30,line width=10pt},
			dfs block span/.style    = {dfs block,line width=3pt,|-|,shorten <=-1pt,shorten >=-1pt},
	]
	
		\def\hobbysize{20pt}
		\node[tree node] (n1) at (1,3.409091) {$1$} ;
		\node[preorder label,above=2pt of n1] {$7$} ;
		\node[subtreesize label,below=2pt of n1] {$0/1$} ;
		\node[tree node] (n2) at (2,4.909091) {$2$} ;
		\node[preorder label,above=2pt of n2] {$6$} ;
		\node[subtreesize label,below=2pt of n2] {$1/3$} ;
		\node[tree node] (n3) at (3,3.409091) {$3$} ;
		\node[preorder label,above=2pt of n3] {$8$} ;
		\node[subtreesize label,below=2pt of n3] {$0/1$} ;
		\node[tree node] (n4) at (4,6.681818) {$4$} ;
		\node[preorder label,above=2pt of n4] {$5$} ;
		\node[subtreesize label,below=2pt of n4] {$3/4$} ;
		\node[tree node] (n5) at (5,8.727273) {$5$} ;
		\node[preorder label,above=2pt of n5] {$4$} ;
		\node[subtreesize label,below=2pt of n5] {$4/8$} ;
		\node[tree node] (n6) at (6,4.909091) {$6$} ;
		\node[preorder label,above=2pt of n6] {$10$} ;
		\node[subtreesize label,below=2pt of n6] {$0/1$} ;
		\node[tree node] (n7) at (7,6.681818) {$7$} ;
		\node[preorder label,above=2pt of n7] {$9$} ;
		\node[subtreesize label,below=2pt of n7] {$1/3$} ;
		\node[tree node] (n8) at (8,4.909091) {$8$} ;
		\node[preorder label,above=2pt of n8] {$11$} ;
		\node[subtreesize label,below=2pt of n8] {$0/1$} ;
		\node[tree node] (n9) at (9,11.045455) {$9$} ;
		\node[preorder label,above=2pt of n9] {$3$} ;
		\node[subtreesize label,below=2pt of n9] {$8/9$} ;
		\node[tree node] (n10) at (10,13.636364) {$10$} ;
		\node[preorder label,above=2pt of n10] {$2$} ;
		\node[subtreesize label,below=2pt of n10] {$9/18$} ;
		\node[tree node] (n11) at (11,3.409091) {$11$} ;
		\node[preorder label,above=2pt of n11] {$16$} ;
		\node[subtreesize label,below=2pt of n11] {$0/1$} ;
		\node[tree node] (n12) at (12,4.909091) {$12$} ;
		\node[preorder label,above=2pt of n12] {$15$} ;
		\node[subtreesize label,below=2pt of n12] {$1/2$} ;
		\node[tree node] (n13) at (13,6.681818) {$13$} ;
		\node[preorder label,above=2pt of n13] {$14$} ;
		\node[subtreesize label,below=2pt of n13] {$2/3$} ;
		\node[tree node] (n14) at (14,8.727273) {$14$} ;
		\node[preorder label,above=2pt of n14] {$13$} ;
		\node[subtreesize label,below=2pt of n14] {$3/5$} ;
		\node[tree node] (n15) at (15,6.681818) {$15$} ;
		\node[preorder label,above=2pt of n15] {$17$} ;
		\node[subtreesize label,below=2pt of n15] {$0/1$} ;
		\node[tree node] (n16) at (16,11.045455) {$16$} ;
		\node[preorder label,above=2pt of n16] {$12$} ;
		\node[subtreesize label,below=2pt of n16] {$5/8$} ;
		\node[tree node] (n17) at (17,6.681818) {$17$} ;
		\node[preorder label,above=2pt of n17] {$19$} ;
		\node[subtreesize label,below=2pt of n17] {$0/1$} ;
		\node[tree node] (n18) at (18,8.727273) {$18$} ;
		\node[preorder label,above=2pt of n18] {$18$} ;
		\node[subtreesize label,below=2pt of n18] {$1/2$} ;
		\node[tree node] (n19) at (19,16.500000) {$19$} ;
		\node[preorder label,above=2pt of n19] {$1$} ;
		\node[subtreesize label,below=2pt of n19] {$18/20$} ;
		\node[tree node] (n20) at (20,13.636364) {$20$} ;
		\node[preorder label,above=2pt of n20] {$20$} ;
		\node[subtreesize label,below=2pt of n20] {$0/1$} ;
	
		\draw[tree edge] (n19) to (n20) ;
		\draw[tree edge] (n19) to (n10) ;
		\draw[tree edge] (n10) to (n16) ;
		\draw[tree edge] (n10) to (n9) ;
		\draw[tree edge] (n9) to (n5) ;
		\draw[tree edge] (n5) to (n7) ;
		\draw[tree edge] (n5) to (n4) ;
		\draw[tree edge] (n4) to (n2) ;
		\draw[tree edge] (n2) to (n3) ;
		\draw[tree edge] (n2) to (n1) ;
		\draw[tree edge] (n7) to (n8) ;
		\draw[tree edge] (n7) to (n6) ;
		\draw[tree edge] (n16) to (n18) ;
		\draw[tree edge] (n16) to (n14) ;
		\draw[tree edge] (n14) to (n15) ;
		\draw[tree edge] (n14) to (n13) ;
		\draw[tree edge] (n13) to (n12) ;
		\draw[tree edge] (n12) to (n11) ;
		\draw[tree edge] (n18) to (n17) ;
	
	%
	%

	\begin{scope}[shift={(0,18)}]
		\foreach \i/\x in {1/20,2/11,3/19,4/8,5/6,6/18,7/14,8/16,9/4,10/3,11/12,12/10,13/9,14/7,15/13,16/5,17/17,18/15,19/1,20/2} {
			\draw[help lines,dashed] (n\i) -- (\i,0.5) ;
		}
		\fill[black!5] (0.5,0) rectangle (20.5,1);
		\draw[black!50,shift={(-.5,0)}] (1,0) grid (21,1) ;
		\foreach \i/\x in {1/20,2/11,3/19,4/8,5/6,6/18,7/14,8/16,9/4,10/3,11/12,12/10,13/9,14/7,15/13,16/5,17/17,18/15,19/1,20/2} {
			\node at (\i,0.5) {\x} ;
			\node at (\i,1.2) {\smaller[2]\ttfamily \i} ;
		}
	\end{scope}
	\end{tikzpicture}
	}}
	\caption{%
		Example of the Cartesian tree $T$ for an array of 20 numbers.
		Each node shows the inorder number (in the node, coincides with the array index), 
		its preorder index (above the node) and the sizes of left subtree and its total subtree 
		(blue, below the node).
		We have $\mathcal H_{\mathit{st}}(T) \approx 28.74$, slightly below the expectation
		$H_{20} \approx 29.2209$.
		The arithmetic code for the preorder sequence of left tree sizes
		is \texttt{111011010111101011110101011111}, \ie, 30 bits.
		This compares very favorably to a typical balanced parenthesis representation
		which would use $40$ bits.
	}
	\label{fig:example-bst}
\end{figure}

The encoding of a tree $T$ stores $n$, the number of nodes followed by 
all left subtree sizes of the nodes in preorder, which we compress 
using \emph{arithmetic coding}~\cite{WittenNealCleary1987}.
To encode $\mathit{ls}(v)$, we feed the arithmetic coder with the model that the next symbol
is a number in $[0..\mathit{st}(v)-1]$ (all equally likely). 
Overall we then need $\mathcal H_{\mathit{st}}(T) + 2$
bits to store $T$ when we know $n$.
(Recall that arithmetic coding compresses to the entropy of the given input
plus at most 2 bits of overhead.)
Taking expectations over the tree $T$ to encode, we can thus store a
binary tree with $n$ nodes using 
$H_n + O(\log n) \approx 1.736n + O(\log n)$ bits on average.

We can reconstruct the tree recursively from this code.
Since we always know the subtree size, we know how many bins the next left tree size uses,
and we know how many nodes we have to read recursively to reconstruct the left and right subtrees.
Since arithmetic coding uses fractional numbers of bits for individual symbols 
(left tree sizes), decoding must always start at the beginning.
This of course precludes any efficient operations on the encoding itself,
so a refined approach is called for.

\subsection{Bounding the worst case}
\label{sec:worst-case}

Although optimal in the average case, the above encoding needs $\Theta(n \log n)$ bits
in the worst case: in a degenerate tree (an $n$-path), 
the subtree sizes are $n,n-1,n-2,\ldots,1$ and hence the 
subtree-size entropy is $\sum_{i=1}^n \lg i \sim n \lg n$.
For such trees, we are better off using one of the standard encodings using $2n+O(1)$
bits on any tree, for example Zak's sequence 
(representing each node by a $1$ and each null pointer by a $0$).
We therefore prefix our encoding with 1 extra bit to indicate the used encoding
and then either
proceed with the average-case optimal or the worst-case optimal encoding.
That combines the best of both worlds with constant extra cost (in time and space).
We can thus encode any tree $t\in\mathcal T_n$ in 
\[2\lceil\lg n\rceil \bin+ \min\bigl\{ \mathcal H_{\mathit{st}}(t)+3, 2n+2 \bigr\} \text{ bits,}\]
where the first summand accounts for storing $n$ (in Elias code).
%
	Since we store $n$ alongside the tree shape, this encoding can be used to
	store any tree with \emph{maximal} size $n$ in the given space.

The above encoding immediately generalizes to, and gives optimal space for, a more general family 
of distributions over tree shapes:
We only need that the distribution of the size of the left subtree of a node $v$ 
depends only on the size of $v$'s subtree
(but is conditionally independent of the shapes of $v$'s subtrees and the shape of the tree above $v$).
In the case of random BSTs, this distribution is uniform (so our encoding is best possible), 
but in general, we can adapt the encoding of our trees to any family of left-size distributions,
in particular to the empirical distributions found in any given tree.
This yields better compression for skewed tree shapes.

\section{Tree Covering}
\label{sec:tree-covering}

In this section, we review the tree-covering (TC) technique
for succinct tree data structures.
The idea of tree covering was introduced by Geary, Raman and Raman~\cite{GearyRamanRaman2006}
(for ordinal trees) and later simplified by Farzan and Munro~\cite{FarzanMunro2012};
He, Munro and Rao~\cite{HeMunroRao2012} added further operations for ordinal trees
and Davoodi et al.~\cite{DavoodiNavarroRamanRao2014} designed inorder support.
Tree covering has predominantly been used for ordinal trees, 
and it can be extended to support updates, but 
our presentation is geared towards succinctly storing static \emph{binary} trees.

\subsection{Trees, mini trees, micro trees}
\label{sec:tc-mini-micro}

The main idea of TC is to partition the nodes of a tree into subsets that form contiguous subtrees.
These subtrees are then (conceptually) contracted into a single node, and a second such grouping
step is applied.
We then represent the overall (binary) tree as a (binary) tree of \emph{mini trees}, 
each of which is a (binary) tree of \emph{micro trees}, each of which is a (binary) tree 
of actual nodes:

\medskip
\medskip
\plaincenter{\includegraphics[width=.6\textwidth]{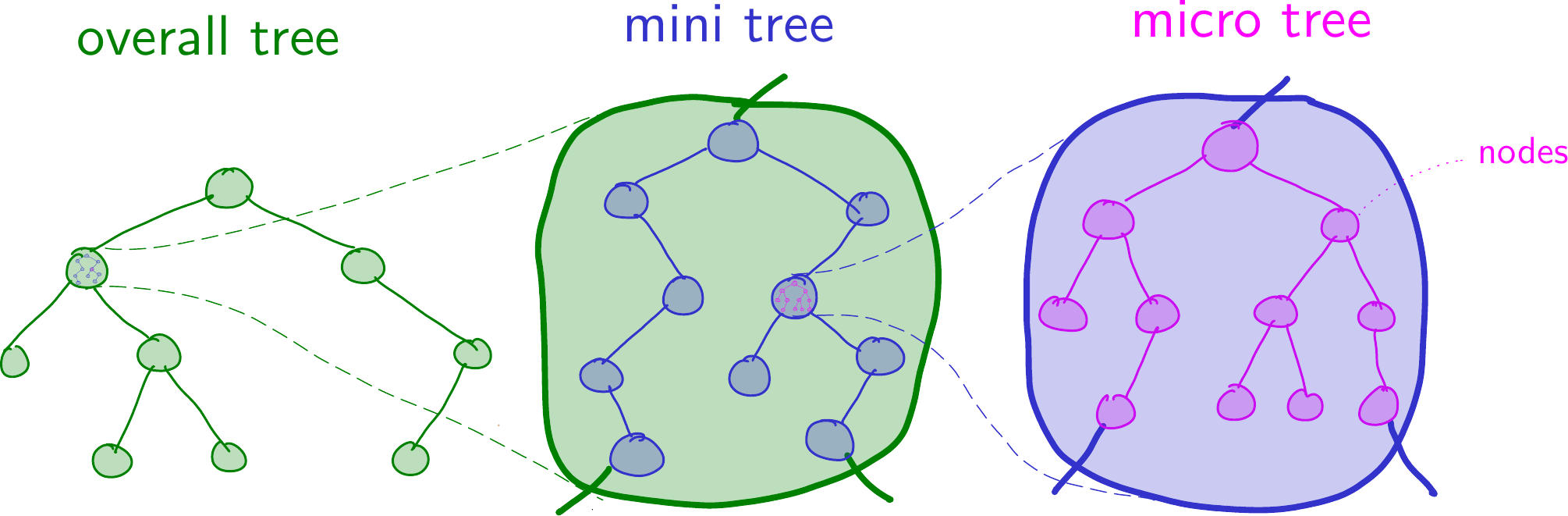}}
\medskip
\medskip

The number of nodes in micro trees is kept small enough so that an encoding
of the local topology of each micro tree fits into less than one word, say $\frac12 \lg n$ bits.
The number of (original) nodes in one mini tree is polylogarithmic (usually $\lg^2 n$),
so that the overall tree (the green level) 
has few enough nodes to fit a non-succinct (\eg, pointer-based) 
representation into $o(n)$ bits of space.
The $O(n/\log^2 n)$ mini trees can similarly be stored with pointers by restricting
them to the polylogarithmic range inside a given mini tree.
More generally, we can afford to store the overall tree (of mini trees, \aka tier-1 macro tree) 
and the mini trees (of micro trees, \aka tier-2 macro trees) 
using any representation that uses $O(n\log n)$ bits for an $n$-node tree.

For the micro trees, we use the so-called ``Four Russian trick'': 
The \emph{type}, \ie, the local topology, for each micro tree
is stored in an array, and used as index into a large precomputed lookup table
to answer any micro-tree-local queries.
Since there are only $2^{\lg n/2} = \sqrt n$ \emph{different} micro-tree types,
this lookup table occupies $o(n)$ bits of space even with precomputed information 
for any pair of nodes.

\subsection{Tree decomposition}
\label{sec:tc-decomposition}

We left open how to find a partition of the nodes into subtrees.
A greedy bottom-up approach suffices to break a tree of $n$ nodes into $O(n/B)$
subtrees of $O(B)$ nodes each~\cite{GearyRamanRaman2006}. 
However, more carefully designed procedures yield additionally
that any subtree has only few edges leading out of this subtree~\cite{FarzanMunro2012}.
While originally formulated for ordinal trees, we note that on binary trees,
the decomposition algorithm of Farzan and Munro~\cite{FarzanMunro2012} 
always creates a \emph{partition} of the nodes,
\ie, there are no subtrees sharing a common root:

\begin{lemma}[{\cite[Theorem~1]{FarzanMunro2012}} for binary trees]
\label{lem:tree-decomposition}
	For any parameter $B\ge 1$, a binary tree with $n$ nodes can be decomposed,
	in linear time, into $\Theta(n/B)$ pairwise disjoint subtrees of $\le 2B$ nodes each.
	Moreover, each of these subtrees has at most three connections to other subtrees:
	\begin{itemize}[itemsep=0pt]
		\item an edge from a parent subtree to the root of the subtree,
		\item an edge to another subtree in the left subtree of the root,
		\item an edge to another subtree in the right subtree of the root.
	\end{itemize}
	In particular, contracting subtrees into single nodes yields again a binary tree.
\end{lemma}

The decomposition scheme of Farzan and Munro additionally guarantees
that (at least) one of the edges to child subtrees emanates from the root;
we do not exploit this property in our data structure, though.

We found that it simplifies the presentation to assume that each subtree
contains a copy of the root of its child subtrees, \ie, all subtree roots
(except for the overall root) are present in two subtrees: once as root of the subtree
and once as a leaf in the parent subtree. We refer to the copy in the parent
as the \emph{portal} of the (parent) subtree (to the child subtree).
(An equivalent point of view is that we partition the \emph{edges} of the tree
and a subtree contains all nodes incident to those edges.)
By \wref{lem:tree-decomposition}, each subtree has at most one left portal and 
one right portal.

\subsection{Node ids: $\tau$-names}
\label{sec:tc-tau-names}

For computations internal to the data structure, we represent nodes by ``$\tau$-names'', 
\ie, by triples $(\tau_1,\tau_2,\tau_3)$, where $\tau_1$ is the (preorder number of the) 
mini-tree, $\tau_2$ is the (mini-tree-local preorder number of the) micro tree
and $\tau_3$ is the (micro-tree-local preorder number of the) node.
Geary, Raman and Raman~\cite{GearyRamanRaman2006} show that the triple fits in $O(1)$ words
and we can find a node's $\tau$-name given its (global) preorder number, and vice versa,
in constant time using additional data structures occupying $o(n)$ bits of space;
we sketch these because they are typical examples of TC index data structures.

\paragraph{preorder $\mapsto$ $\tau$-name: $\mathit{nodeselect}_{\mathit{preorder}}$}

We store $\tau_1$ and $\tau_2$ for all nodes in preorder in two piecewise-constant arrays.
Since mini resp.\ micro trees are complete subtrees except for up to 
two subtrees that are missing,
any tree traversal visits the nodes of one subtree in at most three contiguous ranges,
so the above arrays change their values only $O(n/\log^2 n)$ resp.\ $O(n/\log n)$ times
and can thus be stored in $o(n)$ bits of space by \wref{lem:piecewise-constant-array}.
In a third piecewise-constant array, we store the $\tau_3$-value for the \emph{first} node
in the current micro-tree range;
we obtain the $\tau_3$-value for any given node by adding the result of 
the runlen operation (\wref{rem:piecewise-constant-further-ops}).

\paragraph{$\tau$-name $\mapsto$ preorder: $\mathit{noderank}_{\mathit{preorder}}$}
The inverse mapping is more complicated in~\cite{GearyRamanRaman2006} since subtrees
can have many child subtrees using their decomposition scheme.
With the stronger properties from \wref{lem:tree-decomposition}, computing a node's preorder 
is substantially simpler.
We store for each mini resp.\ micro tree the following information:
\begin{itemize}
\item the preorder index of the root\\
	(in micro trees: local to the containing mini tree),
\item the (subtree-local) preorder index of the left and right portals\\
	(in mini trees: these indices are \wrt actual nodes, not micro trees),
\item the subtree size of the portal's subtrees\\
	(in micro trees: local to the containing mini tree).
\end{itemize}
For mini trees, these numbers fit in $O(\log n)$ bits each and can thus be stored for all
$O(n/\log^2n)$ mini trees in $o(n)$ space.
For micro trees, we store information relative to the surrounding mini tree, so the numbers
fit in $O(\log \log n$) bits each, again yielding $o(n)$ extra space overall.

To compute the preorder index for node $v$ with $\tau$-name $(\tau_1,\tau_2,\tau_3)$,
we start with the (mini-tree-local) preorder index of $\tau_2$'s root within $\tau_1$
and we add the subtree size of $\tau_2$'s left (right) portal if $\tau_3$ is larger 
than the left (right) portal's micro-tree-local preorder index.
The result is the preorder index $\tau_{3}'$ of $v$ within $\tau_1$ 
(in terms of nodes, not in terms of micro trees).
Applying the same computation, but using $\tau_1$'s portal information and starting with $\tau_3'$,
we obtain $v$'s global preorder index.

\subsection{Lowest common ancestors}
\label{sec:tc-lca}

He, Munro and Rao~\cite{HeMunroRao2012} use known $O(n\log n)$-bit-space 
representations 
of ordinal trees that support constant-time lowest-common-ancestor queries to represent
the tier-1 and tier-2 macro trees, and they give a (somewhat intricate) $O(1)$-time algorithm
to compute lca based on these representations.

A conceptually simpler option is to use a $O(n)$-bit 
data structure for (static, ordinal) trees that supports lca queries by preorder indices,
\eg, the one of Navarro and Sadakane~\cite{NavarroSadakane2014}.
Note that for lca queries and preorder numbers, the distinction in binary trees 
between unary nodes with only a left resp.\ only a right child is irrelevant.
We construct the tree $T_B$ of all micro-tree roots by contracting micro trees into single nodes,
(ignoring the mini-tree level),
and store $T_B$ using the compact ordinal-tree data structure.

Using similar data structures as for $\mathit{noderank}_{\mathit{preorder}}$, 
we can compute the preorder number of a micro-tree root in the micro-tree-root tree $T_B$, 
denoted by $\mathit{noderank}_{\mathit{microRoot}}(\tau_1,\tau_2)$.
The inverse operation, $\mathit{nodeselect}_{\mathit{microRoot}}$,
is similar to $\mathit{nodeselect}_\mathit{preorder}$.

To find the LCA of two nodes $u$ and $v$ in the tree,
we first find
$k_u = \mathit{noderank}_{\mathit{microRoot}}(u)$ and
$k_v = \mathit{noderank}_{\mathit{microRoot}}(v)$.
If $k_u = k_v$, both nodes are in the same micro tree and so is their LCA,
so we use the (micro-tree-local) lookup table to find the (precomputed) LCA.

Otherwise, when $k_u \ne k_v$, we determine $k = \mathit{lca}(T_B, k_u, k_v)$, the
LCA of the micro-tree roots in $T_B$.
If $k_u \ne k \ne k_v$, then the paths from $k_u$ and $k_v$ first meet at micro-tree root $k$ 
and we return $\mathit{lca}(u,v) = \mathit{nodeselect}_\mathit{microRoot}(k)$.

The remaining case is that the LCA is one of $k_u$ or $k_v$; 
\withoutlossofgenerality say $k = k_u$.
This means that the micro-tree root $r = \mathit{nodeselect}_\mathit{microRoot}(k_u)$ 
of $u$'s micro-tree root is an ancestor of~$v$. 
\droppedForArXiv{
Note that this is not necessarily the case for $u$ and $v$ itself, 
consider for example nodes with preorder index $22$ and $26$ in \wref{fig:blocks-example}.
}
Since $v$ is not in $r$'s micro tree, but in $r$'s subtree, 
it must be in the subtree of one of the portals.
We find the right one 
by selecting the portal $p$ that fulfills 
$p = \mathit{lca}(T_B,p,v)$.
(If $T_B$ supports ancestor queries more efficiently, 
we can also ask whether $p$ is an ancestor of~$v$.)
Finally, as in the intra-micro-tree case, we find the lca of $u$ and $p$ within the micro tree.

\subsection{Inorder rank and select}
\label{sec:tc-inorder}

Davoodi et al.~\cite{DavoodiNavarroRamanRao2014} describe $o(n)$ bit data structures 
that allow to map preorder to inorder and vice versa.
For computing the inorder of a node $v$, they use the equation
$\mathit{inorder}(v) = \mathit{preorder}(v) + \mathit{ls}(v) - \mathit{rightdepth}(v)$,
where $\mathit{rightdepth}(v)$ is the the number of right-edges on the path from the root to $v$,
\ie, the depth where only following right-child pointers counts.%
\footnote{
	Davoodi et al.\ use Ldepth for this because they consider the path \emph{from $v$ to the root}.
	We found ``right depth'' (corresponding to the direction from root to node) more compatible
	with widespread convention.
}
Storing the global right depth of each mini-tree root, the mini-tree-local right depth of 
each micro-tree root and micro-tree-local right-depth info in the lookup table,
we can obtain the right-depth of any node by adding up these three quantities.

Instead of directly mapping from inorder to preorder (as done in~\cite{DavoodiNavarroRamanRao2014}),
we can adapt the above strategy for mapping \emph{inorder} to $\tau$-name.
Storing $\tau_1$ and $\tau_2$ for all nodes in inorder works the same as above.
We cannot store $\tau_3$ directly in a piecewise constant array (it changes to often), 
but we can store the micro-tree-local \emph{inorder} index, $\tau_3^{\textrm{in}}$.
Finally, we use the lookup table to translate $\tau_3^{\textrm{in}}$ to $\tau_3$.

\section{Entropy trees}
\label{sec:entropy-trees}

The dominating contribution (in terms of space) in TC comes from the array of micro-tree types.
All other data structures~-- for storing overall tree and mini trees, 
as well as the various index data structures to support queries~-- fit in $o(n)$ bits of space.
Encoding micro trees using one of the $2n$-bit encodings for $\mathcal T_n$
(\eg, BP, DFUDS, or Zaks' sequence),
the combined space usage of the types of all micro trees is $2n+o(n)$ bits.

The micro-tree types are solely used as an index for the lookup table;
how exactly the tree topology is encoded is immaterial.
It is therefore possible to replace the fixed-length encoding by one that
adapts to the actual input.
Using variable-cell arrays (\wref{lem:variable-cell-arrays}), 
the dominant space is the sum of the lengths of the micro-tree types.

Davoodi et al.~\cite{DavoodiNavarroRamanRao2014} used the ``ultra-succinct'' encoding proposed by
Jansson, Sadakane and Sung~\cite{JanssonSadakaneSung2012} for micro trees to obtain
a data structure that adapts to the entropy of the node-degree distribution.
This approach is inherently limited to non-optimal compression since it only depends on
the local order of a fixed number of values in the input array.
Using our new encoding for binary tree we can overcome this limitation.

Combining the tree-covering data structure for binary trees described in \wref{sec:tree-covering}
with the (length bounded) subtree-size code for binary trees from \wref{sec:encoding}
yields the following result.

\begin{theorem}[Entropy trees]
	Let $t\in\mathcal T_n$ be a (static) binary tree on $n$ nodes.
	There is a data structure that occupies $\min\{\mathcal H_{\mathit{st}}(t),2n\} + o(n)$ 
	bits of space and 
	supports the following operations in $\Oh(1)$ time:
	\begin{itemize}[itemsep=0pt]
	\item find the node with given pre- or inorder index,
	\item find the pre- and inorder index of a given node,
	\item compute the lca of two given nodes.
	\end{itemize}
\end{theorem}

\begin{proof}
The correctness of operations and size of supporting data structures directly 
following from the previous work on tree covering (see \wref{sec:tree-covering})
and the discussion above.
It remains to bound the sum of code lengths of micro-tree types.
Let $\mu_1,\ldots,\mu_\ell$ denote the micro trees resulting from the tree decomposition.
We have
\begin{align*}
		\sum_{i=1}^\ell |\mathit{type}(\mu_i)|
	&\wwrel\le
		\sum_{i=1}^\ell \sum_{v\in\mu_i} \Bigl(\mathit{st}_{\mu_i}(v)+2\Bigr)
\\	&\wwrel\le
		\sum_{i=1}^\ell \sum_{v\in\mu_i} \mathit{st}_{t}(v) 
		\wwbin+ \Oh\biggl(\frac n{\log n}\biggr)
\\	&\wwrel\le
		\sum_{v\in t} \mathit{st}_{t}(v) 
		\wwbin+ \Oh\biggl(\frac n{\log n}\biggr)
\\	&\wwrel=
		\mathcal H_{\mathit{st}}(t)
		\wwbin+ \Oh\biggl(\frac n{\log n}\biggr).
\end{align*}
Moreover, $|\mathit{type}(\mu_i)| \le 2 |\mu_i|+\Oh(1)$, so also 
$\sum_{i=1}^\ell |\mathit{type}(\mu_i)| \le 2n + \Oh\bigl(\frac n{\log n}\bigr)$.
\end{proof}

By using this data structure on the Cartesian tree of an array,
we obtain a compressed RMQ data structure.

\begin{corollary}[Average-case optimal succinct RMQ]
	There is a data structure that supports (static) range-minimum queries on an array $A$ of 
	$n$ (distinct) numbers in $\Oh(1)$ worst-case time and which
	occupies $H_n + o(n) \approx 1.736 n + o(n)$
	bits of space on average over all possible permutations of the elements in $A$.
	The worst case space usage is $2n + o(n)$ bits.
\end{corollary}

\paragraph{Further operations}

Since micro-tree types are only used for the lookup table, 
all previously described index data structures for other operations are not affected by
swapping out the micro-tree encoding.
Entropy trees can thus support the full range of (cardinal-tree) operations
listed in \cite[Table~2]{FarzanMunro2012}.

\section{Hyper-succinct trees}
\label{sec:hyper-succinct-trees}

The subtree-size code yields optimal compression for random BSTs, 
but is not a good choice for certain other shape distributions.
However, we obtain instance-optimal compression by yet another encoding for micro trees:
By treating each occurring micro-tree type as a single symbol and counting how often it occurs 
when storing a given tree $t$, we can compute a (length-bounded) Huffman code
for the micro-tree types.
We call the resulting tree-covering data structure \emph{``hyper succinct''}
since it yields better compression than ultra-succinct trees.
Indeed, it follows from the optimality of Huffman codes that no other tree-covering-based
data structure can use less space.

It is, however, quite unclear how good the compression for a given tree is;
note in particular that the set of micro trees is a mixture of subtrees
at the fringe of $t$ and subtrees from the ``middle'' of $t$ where one or two 
large subtrees have been pruned away.
How repetitive this set of shapes is will depend not only on $t$, but also on the
tree-decomposition algorithm and the size of micro trees.

\bibliographystyle{plainurl}
\bibliography{refs}

\end{document}